\documentclass[11pt]{article}
\usepackage{geometry} 
\usepackage[utf8]{inputenc}
\usepackage{amsmath,amssymb}
\usepackage{color}
\usepackage{calc}
\usepackage{tikz}
\usetikzlibrary{positioning,arrows,decorations.pathreplacing,shapes}
 \usetikzlibrary{calc}
\usepackage{multirow}
\usepackage{boxedminipage}
 \usepackage{xifthen}
 \usepackage{tabularx}

\usepackage{algorithmic}
\usepackage[linesnumbered,vlined,ruled]{algorithm2e}
\usepackage{hyperref}
\newcommand{\footremember}[2]{%
    \footnote{#2}
    \newcounter{#1}
    \setcounter{#1}{\value{footnote}}%
}
\newcommand{\footrecall}[1]{%
    \footnotemark[\value{#1}]%
}

\newcommand{\cO}{\mathcal{O}}

\newcommand{\reject}{\textsf{reject}}
\newcommand{\GRECON}{$\mathcal{G}$-\textsc{Recognition}}
\newcommand{\GSRECO}{$\mathcal{G}$-\textsc{Strong-Rec}}
\newcommand{\GRECO}{$\mathcal{G}$-\textsc{Weak-Rec}}

\newcommand{\GSRECOP}{\ProblemName{$\mathcal{G}$-\textsc{Strong-Rec}}}
\newcommand{\GRECOP}{\ProblemName{$\mathcal{G}$-\textsc{Weak-Rec}}}

\newtheorem{theorem}{Theorem}
\newtheorem{corollary}{Corollary}
\newtheorem{definition}{Definition}
\newtheorem{lemma}{Lemma}
\newtheorem{remark}{Remark}
\newenvironment{proof}[1][]{\par \noindent {\bf Proof #1}\ }{\hfill$\Box$\par \vspace{11pt}}

\newcommand{\pname}{\textsc}
\newcommand{\ProblemFormat}[1]{\pname{#1}}
\newcommand{\ProblemIndex}[1]{\index{problem!\ProblemFormat{#1}}}
\newcommand{\ProblemName}[1]{\ProblemFormat{#1}\ProblemIndex{#1}{}\xspace}

\newlength{\RoundedBoxWidth}
\newsavebox{\GrayRoundedBox}
\newenvironment{GrayBox}[1]%
   {\setlength{\RoundedBoxWidth}{.93\textwidth}
    \def\boxheading{#1}
    \begin{lrbox}{\GrayRoundedBox}
       \begin{minipage}{\RoundedBoxWidth}}%
   {   \end{minipage}
    \end{lrbox}
    \begin{center}
    \begin{tikzpicture}%
       \node(Text)[draw=black!80,fill=white,rounded corners,%
             inner sep=2ex,text width=\RoundedBoxWidth]%
             {\usebox{\GrayRoundedBox}};
        \coordinate(x) at (current bounding box.north west);
        \node [draw=white,rectangle,inner sep=3pt,anchor=north west,fill=white] 
        at ($(x)+(6pt,.75em)$) {\boxheading};
    \end{tikzpicture}
    \end{center}}     

\newenvironment{defproblemx}[2][]{\noindent\ignorespaces%
                                \FrameSep=6pt%
                                \parindent=0pt%
                \vspace*{-1.5em}
                \ifthenelse{\isempty{#1}}{%
                  \begin{GrayBox}{\textsc{#2}}%
                }{%
                  \begin{GrayBox}{\textsc{#2} parameterized by~{#1}}%
                }
                \begin{tabular*}{\textwidth}{@{\hspace{.1em}} >{\itshape} p{1.8cm} p{0.8\textwidth} @{}}%
            }{
                \end{tabular*}%
                \end{GrayBox}%
                \ignorespacesafterend
            }  

\newcommand{\defproblema}[3]{
  \begin{defproblemx}{#1}
    Input:  & #2 \\
    Output: & #3
  \end{defproblemx}
}%
\usepackage{rotating}
\usepackage{framed}

\title{Graph Reconstruction in the Congested Clique}

\author{%
Pedro Montealegre\footremember{1}{Facultad de Ingenier\'{\i}a y Ciencias, Univ. Adolfo Ib\'a\~nez, Santiago, Chile, \texttt{pedro.montealegre@uai.cl}}
\and Sebastian Perez-Salazar\footremember{2}{DIM-CMM  (UMI 2807 CNRS), Univ. de Chile, Santiago, Chile, \texttt{\{sperez,rapaport\}@dim.uchile.cl}}
\and Ivan Rapaport\footrecall{2}
\and Ioan Todinca\footremember{3}{Univ. Orl\'{e}ans, INSA Centre Val de Loire, LIFO EA 4022, Orl\'{e}ans, France, \texttt{ioan.todinca@univ-orleans.fr}}
}

\begin{document}

\maketitle

\vspace{-0.5cm}
\begin{abstract}

The congested clique model is a message-passing model of distributed computation where the underlying
communication network is the complete graph of $n$ nodes. In this paper we consider the situation where the joint input to
the nodes is an $n$-node labeled graph $G$, i.e., the local input received by each node is the indicator function
of its neighborhood in $G$. Nodes execute an algorithm, communicating with each other
in synchronous rounds and their goal is to compute some function that depends on $G$.
In every round, each of the $n$ nodes may send up to $n-1$ different $b$-bit messages through each of its $n-1$
communication links. We denote by $R$ the number of rounds of the algorithm. The product $Rb$, that is, the total number of bits received by a node through one link, is the {\emph{cost}}
of the algorithm.

The most difficult problem we could attempt to solve is the {\emph{reconstruction problem}},
where nodes are asked to recover all the edges of the input graph $G$. Formally, given a class of graphs $\mathcal G$, the problem is defined as
follows: if $G \notin {\mathcal G}$, then every node must \reject; on the other hand, if $G \in {\mathcal G}$, then every node must end up, after the $R$ rounds,
knowing all the edges of $G$. It is not difficult to see that the cost $Rb$ of any algorithm that solves this problem (even with public coins)
is at least $\Omega(\log|\mathcal{G}_n|/n)$, where $\mathcal{G}_n$ is the subclass of all $n$-node labeled graphs in $\mathcal G$. In this paper we prove that
previous bound is tight and that it is possible to achieve it with only $R=2$ rounds. More precisely, we exhibit (i) a one-round algorithm that achieves this bound for hereditary graph classes;  and (ii) a two-round algorithm that achieves this bound for arbitrary graph classes. Later, we show that the bound $\Omega(\log|\mathcal{G}_n|/n)$ cannot be achieved in one-round for arbitrary graph classes, and we give tight algorithms for that case. 

From (i)  we recover all known results concerning the reconstruction of graph classes in one round and bandwidth $\cO(\log n)$: forests, planar graphs, cographs, etc. But we also get new one-round algorithms for other hereditary graph classes such as unit disc graphs, interval graphs, etc. From (ii), we can conclude that any problem restricted to a class of graphs of size  $2^{\cO(n\log n)}$ can be solved in the congested clique model in two rounds, with bandwidth $\cO(\log n)$. Moreover, our general two-round algorithm is valid for any set of labeled graphs, not only for graph classes (which are sets of labeled graphs closed under isomorphims).

\end{abstract}

\section{Introduction}

The  {\emph{congested clique}} model 
--a message-passing model of distributed
computation where the underlying communication network is the complete graph~\cite{lotker2005minimum}--
is receiving increasingly more attention~\cite{berns2012super,Censor-HillelKK15,dolev2012,drucker2014,ghaffari2016improved,ghaffari2016mst,hegeman2015toward,hegeman2014b,lenzen2013,patt2011}. There are deep connections between the  congested clique model and popular 
distributed systems 
such as the $k$-machine model~\cite{klauck2015distributed} or MapReduce~\cite{hegeman2014}.  Moreover, with the emergence of large-scale networks, this  model  has started to be used 
in other areas such as distributed convex learning~\cite{arjevani2015communication}.

The  congested clique model is defined as follows. 
There are $n$ nodes which are given distinct identities (IDs), that we assume for simplicity to be numbers between 1 and $n$.
In this paper we consider the situation where the joint input to the nodes {\emph{is a graph}} $G$. More precisely, each node $v$ receives 
as input an $n$-bit boolean vector  $x_v \in \{0,1\}^n$,  which is the indicator function 
of its neighborhood in  $G$. Note that the 
 input graph $G$ is an arbitrary $n$-node graph,
{\emph{a subgraph of the communication network  $K_n$}}.

Nodes execute an algorithm, communicating
with each other in synchronous rounds and their goal is to compute  some function $f$  that depends on $G$. In every round, each of the $n$ nodes may send up to $n-1$ different $b$-bit messages through  each of its  $n-1$ communication links.  When an algorithm  stops {\emph{every node must know}} $f(G)$. We call $f(G)$ the  {\emph{output}} of the distributed algorithm. The parameter $b$ is known as
the {\emph{bandwidth}} of the algorithm. We denote by $R$ the \emph{number of rounds}. The product $Rb$ represents the total number of bits received by a node through one link, and we call it the \emph{cost} of the algorithm.

An algorithm may  be deterministic or randomized. We distinguish two sub-cases of randomized algorithms: the private-coin setting, where each node flips its own coin; and the public-coin setting, where the coin is shared between all nodes. An  $\varepsilon$-error algorithm ${\mathcal A}$ that computes a function $f$ is a randomized algorithm
 such that, for every input graph $G$, $\Pr( \mathcal{A} \textrm{ outputs } f(G) )\geq 1- \varepsilon$. In the case where $\varepsilon \to 0$ as $n \to \infty$, we say that ${\mathcal A}$ computes $f$ with high probability (whp). 

Function $f$ defines the problem to be solved. A $0-1$ function corresponds to a decision problem (such as 
connectivity~\cite{hegeman2015toward}). For other, more general types of problems, $f$ should be defined, in fact, 
as a relation. This happens, for instance, when we want to construct a minimum spanning tree~\cite{ghaffari2016mst},  
a 3-ruling set~\cite{hegeman2014b},  all-pairs shortest-paths~\cite{Censor-HillelKK15}, etc. 

The most difficult problem we could attempt to solve is the  {\emph{reconstruction problem}}, where nodes are 
asked to reconstruct the input graph $G$. 
In fact, if at the end of the algorithm every node $v$ has full knowledge  of $G$, then it could answer  any question concerning $G$. 
(This holds because in the congested clique model nodes have unbounded computational power and the only cost is related to communication).

In centralized, classical graph algorithms,  a widely used approach to cope with NP-hardness is to restrict the class of graphs where the input $G$ belongs.
Consider, for instance, the coloring problem, where the goal is to determine
the minimum number of colors that we can assign to the vertices of $G$ such that  no two vertices sharing the same edge have the same 
color~\cite{garey2002computers}.
It is known that, if the input is restricted to the class $\mathcal G$ of interval graphs, the coloring problem is 
polynomial~\cite{garey2002computers}. 
Nevertheless, if we restrict it to
planar graphs, the problem remains NP-complete~\cite{garey2002computers}. We are going to use the same approach here, in the congested clique model. But, as we are going to explain later, surprisingly, the complexity of the reconstruction problem {\emph{will only depend on the cardinality}} of the subclass of $n$-node graphs in    $\mathcal G$.


Formally, for any fixed set of graphs  $\mathcal G$ we 
are going to introduce  two problems.
The first one, the  {\emph{strong recognition problem}} {\GSRECO}, is the following.

\bigskip
%
%
%
%
%

\defproblema{\GSRECOP}%
{An arbitrary graph  $G$}%
{$\begin{cases}
   \text{all the edges of } G & \text{if } G \in {\mathcal G};\\
   \mbox{reject } & \text{otherwise.}
  \end{cases}$}

\medskip

 Recall that the output is computed by {\emph{every node}} of the network. In other words, every node of an algorithm that solves  {\GSRECO}
must end up knowing whether $G$ belongs to ${\mathcal G}$; and, in the positive cases, every node also finishes
knowing all the edges of $G$. Note that, in principle, $\mathcal G$ could be defined as the set of all graphs.

We also define a  {\emph{weak recognition problem}} {\GRECO}. This  is a promise problem, where the input graph $G$ is promised to belong to ${\mathcal G}$.
In other words, for graphs that do not belong to ${\mathcal{G}}$, the behavior of an algorithm that solves {\GRECO} does not matter.

\medskip
%
%
%

\defproblema{\GRECOP}%
{$G \in {\mathcal G}$}%
{all the edges of $G$}

\medskip

For any positive integer $n$ we define ${\mathcal G}_n$ as the set of $n$-node graphs in $\mathcal G$.
There is an obvious lower bound for $Rb$, even for the  weak reconstruction problem  {\GRECO} and even in the  public-coin setting. In fact, 
$Rb= \Omega(\log|\mathcal{G}_n|/n)$. This can be easily seen if we note that, in 
the randomized case, there must be at least one outcome of the coin tosses for which the correct algorithm reconstructs the input graph  in at least  $(1-{\varepsilon})$ of the cases.  Therefore, $n+(n-1)Rb=\Omega((1-{\varepsilon})\log|\mathcal{G}_n|)=\Omega(\log|\mathcal{G}_n|)$. The value $(n-1)Rb+n$ corresponds to the total number of bits received by any 
node $v$ of the network: $(n-1)Rb$ bits are received from the other nodes and $n$ bits are known by $v$  at the beginning of the algorithm
(this is the indicator function of its neighborhood). This implies that
$Rb=\Omega(\log|\mathcal{G}_n|/n)$. In this paper we are going to prove that this bound is tight even 
with $R=1$ (if $\mathcal G$ is an hereditary class of graphs) and $R=2$ (in the general case).

We point
out that our reconstruction algorithms may be applied not only to $G$ itself but also  to some subgraph of $G$. For instance, consider the situation where we generate a new graph $H$ by performing (locally) a random sampling on the edges of $G$. Since $H$ typically belongs  to a smaller class of graphs (whp), reconstructing $H$ may result in an efficient strategy to infer some properties of $G$~\cite{spielman2011spectral}.


\subsection{Our Results}

We start this paper by  studying a very natural family 
of graph classes known as  \emph{hereditary}. A class ${\mathcal G}$ is hereditary if,  for every graph $G \in \mathcal G$, every induced subgraph of $G$ also belongs to $\mathcal G$.  Many  graph classes  are hereditary: forests, planar graphs, bipartite graphs, $k$-colorable graphs,  
bounded tree-width graphs, $d$-degenerate graphs, etc.~\cite{brandstadt1999graph}. Moreover, any intersection class of graphs --such as interval graphs, chordal graphs, unit disc graphs, etc.-- is also
hereditary~\cite{brandstadt1999graph}.

In Section~\ref{sec:hereditary} we give, for every hereditary class of graphs $\mathcal{G}$, a one-round private-coin randomized algorithm 
that solves  {\GSRECO} with bandwidth $$\cO( \max_{k \in [n]} {\log |\mathcal{G}_k|}/{k} + \log n).$$
We emphasize that our algorithm runs in one-round, and therefore it runs in the \emph{broadcast congested clique}, 
a restricted version  of the congested clique model where, in every round, the $n-1$ messages sent by a node must be the same. (This equivalence will be explained in Section \ref{sec:prelim}). 
We also remark that for many hereditary graph classes, including all classes listed above, our algorithm is tight. Moreover, our result implies that {\GSRECO} can be solved in one-round with bandwidth $\cO(\log n)$ when $\mathcal{G}$ is the class of forests, planar graphs, interval graphs, unit-circle graphs, or any other hereditary graph class $\mathcal G$ such that 
$|{\mathcal G}_n| = 2^{\cO(n\log n)}$.

In Section  \ref{sec:general} we give a very general result, showing that two rounds are sufficient to solve {\GSRECO} in the congested clique model, for any set of graphs $\mathcal{G}$. More precisely, we provide a two-round deterministic algorithm that solves {\GRECO} and a two-round private-coin randomized algorithm that solves {\GSRECO} whp. We also give a three-round deterministic algorithm solving {\GSRECO}. All algorithms run using bandwidth $\cO( \log |\mathcal{G}_n|/n + \log n)$, so they are asymptotically optimal when $|\mathcal{G}_n| = 2^{\Omega(n\log n)}$.  

Our result implies, in particular, that {\GSRECO} can be solved in two rounds with bandwidth $\cO(\log n)$, when $\mathcal{G}$ is  any set of graphs of size $2^{\cO(n\log n)}$. The only property 
of the set of graphs $\mathcal{G}$  used by our algorithm is the cardinality of $\mathcal{G}_n$. Our algorithm
does not require  $\mathcal{G}$ to be closed under isomorphisms.

In Section \ref{sec:revisiting} we revisit the one-round case. We show that our general algorithm can be adapted to run in one round (i.e., in the broadcast congested clique model) by allowing a larger bandwidth, and then we show that this is tight. More precisely, we show that, for every set of graphs $\mathcal{G}$,  there is a one-round deterministic algorithm that solves {\GRECO}, and a one-round private-coin algorithm that solves  {\GSRECO} whp,  both of them using 
bandwidth  $\cO(\sqrt{ \log |\mathcal{G}_n| \log n} + \log n)$. 

Then we show that there are classes of graphs $\mathcal{G}$ satisfying that $|\mathcal{G}_n| \leq 2^{\cO(n)}$ such that every algorithm (deterministic or randomized) that solves  {\GRECO} in the broadcast congested clique model has cost $Rb=\Omega(\sqrt{\log |\mathcal{G}_n|})$. Therefore, with respect to the bandwidth, our general one-round algorithms for solving {\GRECO} and {\GSRECO} are tight (up to a logarithmic factor).

Our one-round algorithm that solves  {\GSRECO}  uses private coins. Is it possible to achieve the
same deterministically? Our last result gives a negative answer to this question. Consider, 
for a set of graphs $\mathcal G$,  the \emph{recognition problem} {\GRECON}, which consists in deciding whether the input graph $G$ belongs to $\mathcal{G}$. We show that there exists a set of graphs $\mathcal{S}$, satisfying $|\mathcal{S}_n| \leq 2^{n}$, such that any one-round deterministic algorithm that solves $\mathcal{S}$-{\sc Recognition}
 requires bandwidth $\Omega(n) = \Omega(\log |\mathcal{S}_n|)$. Clearly, the same lower-bound is valid for any deterministic algorithm that solves
 $\mathcal{S}$-{\sc Strong-Rec}. This is far from our bandwidth
 $\cO(\sqrt{n\log n})=\cO(\sqrt{ \log |\mathcal{G}_n| \log n} + \log n)$.

\subsection{Related Work}

All known results concerning the reconstruction of graphs have been obtained in the context of hereditary graph classes.
For instance,  let $\mathcal{G}$ be the class of {\emph{cograph}}, that is, the class of graphs that do not contain the 4-node path as an induced subgraph. This class is obviously hereditary. In~\cite{kari2015}, the authors presented a one-round public-coin algorithm that solves {\GSRECO} with bandwidth $\cO(\log n)$.  Note that $|{\mathcal G}_n| =\Theta(2^{n\log n})$. Therefore,  the result we get in this paper is stronger, because our one-round algorithm needs the same bandwidth but uses private coins.

In~\cite{BMN+11,MoTo16} it is shown that, if $\mathcal{G}$ is the class of \emph{$d$-degenerate} graphs, then there is a one-round deterministic algorithm that solves {\GSRECO} with bandwidth $\cO(d\log n)=\cO(\log n)$. A graph $G$ is $d$-degenerate if one can remove from $G$ a vertex $r$ of degree at most $d$, and then proceed recursively on the resulting graph $G' = G-r$, until obtaining the empty graph. Note that planar graphs (or more generally, bounded genus graphs), bounded tree-width graphs, graphs without a fixed graph $H$ as a minor, are all $d$-degenerate, for some constant $d>0$. Since  the class of $d$-degenerate graphs is hereditary and satisfies  $|{\mathcal G}_n| =\Theta(2^{n\log n})$, it follows, from this paper, the existence of a one-round private-coin randomized algorithm 
that solves  {\GSRECO} with bandwidth $\cO(\log n)$. However, the result of~\cite{BMN+11} for this particular class is stronger, since their algorithm is deterministic.

Another example of reconstruction with one-round algorithms  can be found in~\cite{drucker2014}. There, 
the authors consider the class of graphs defined by one forbidden subgraph $H$. They show that such classes  can be reconstructed deterministically 
with cost $Rb=\cO((ex(n,H)\log n)/n)$, where $ex(n,H)$ is the {\emph{Tur\'an number}}, defined as be the maximum number of edges in an $n$-node  graph  not containing an isomorphic copy of $H$ as a subgraph. For example, if $C_{4}$ is the cycle of length 4, then $ex(n,C_{4}) = \cO(n^{3/2})$. This implies that, if we define ${\mathcal G}$ as the class of graphs not containing $C_{4}$ as a subgraph, then there is a one-round deterministic algorithm that solves {\GSRECO}  with bandwidth $\cO(\sqrt{n}\log n)$. 

\section{Preliminaries}\label{sec:prelim}

\subsection{Some Graph Terminology}

Two graphs $G$ and $H$ are \emph{isomorphic} if there exists a bijection $\varphi: V(G) \rightarrow V(H)$ such that any pair of vertices $u,v$ are adjacent in $G$ if and only if $f(u)$ and $f(v)$ are adjacent in $H$. 
A {\emph{class of graphs}} $\mathcal{G}$ is a set of graphs which is {\it closed under  isomorphisms}, i.e., if $G$ belongs to $\mathcal{G}$ and $H$ is isomorphic to $G$, then $H$ also belongs to $\mathcal{G}$.  For a class of graphs $\mathcal{G}$ and $n>0$, we call $\mathcal{G}_n$ the subclass of $n$-node graphs  in $\mathcal{G}$.

For a graph $G=(V,E)$ and $U \subseteq V$ we denote $G[U]$ the subgraph of $G$ induced by $U$. More precisely,
the vertex set of $G[U]$ is $U$ and  the edge set consists of all of the edges in $E$ that have both endpoints in $U$.
A class of graphs $\mathcal{G}$ is \emph{hereditary} if it is closed under taking induced subgraphs, i.e.,  for every  $G=(V,E) \in \mathcal{G}$ and every $U \subseteq V$, the induced subgraph $G[U] \in \mathcal{G}$.

For a graph $G = (\{v_1, \dots, v_n\}, E)$, we call $A(G)$ its \emph{adjacency matrix}, i.e.,  the 0-1  square matrix of dimension $n$ where $[A(G)]_{ij}=1$ if and only if $v_i$ is adjacent to $v_j$. Let $M$ be a square matrix of dimension $n$, and let $i\in [n]=\{1,\ldots,n\}$. We call $M_i$ the $i$-th row of $M$. Let $N$  be another square matrix of dimension $n$. We denote by $d_r(M, N)$ the \emph{row-distance} between $M$ and $N$, that is, the number of rows that are different between $M$ and $N$. In other words, $d_r(M,N) = \{i \in [n] : M_i \neq N_i\}$. For $k>0$ and $G=(V,E)$, let us call $B(G,k)$ the set of all graphs $H = (V, E')$ such that $d_v(A(G), A(H)) = k$.

\subsection{One-Round Algorithms in the Congested Clique}

The \emph{broadcast congested clique} is a restricted version of the congested clique model where each node is forced, 
in each round, to send the same message through its $n-1$ communication links. But, if we consider one-round algorithms, the two models
are the same. In fact, suppose that there is a one-round algorithm $\mathcal{A}$ (deterministic or randomized) in the congested clique with bandwidth $b$. We can transform it into an algorithm $\mathcal{B}$ in the broadcast version with bandwidth $b+1$ as follows.
We fix a vertex, say the one with ID $1$, and every node $j$ broadcasts the message it would send to node $1$ on algorithm $\mathcal{A}$, plus one bit indicating whether node $j$ and node $1$ are adjacent in $G$. After this communication round of $\mathcal{B}$, every node knows the messages node $1$ would have received after the communication round of algorithm $\mathcal{A}$. 
Moreover, every node knows the neighborhood of node $1$. The result follows from the fact that, with this information, node 1 knows the output.
Obviously,  as we will see in this paper, when multi-round algorithms are considered, the broadcast congested clique model 
is much less powerful than the congested clique model. 

\subsection{Fingerprints}\label{subsec:fin}

The following technique, that we call \emph{fingerprints}, is based on a result known as the Schwartz Zippel Lemma, used in verification of polynomial identities~\cite{schwartz1980fast}. Let $n$ be a positive integer and $p$ be a prime number.  In the following, we denote by $\mathbb{F}_p$ the finite field of size $p$ (we refer to the book of  Lidl and Niederreiter \cite{lidl1994introduction} for further details and definitions involving finite fields). A polynomial $P \in \mathbb{F}_p[X]$ of \emph{degree} $d$ is an expression of the form $P(x) = \sum_{i=0}^d a_i x^i$, where $a_i \in \mathbb{F}_p$ and $a_i\neq 0$ for each $0 \leq i \leq d$. We denote by $\mathbb{F}_p[X]$ the polynomial ring on $\mathbb{F}_p$.  An element $b \in \mathbb{F}_p$ is called a \emph{root} of a polynomial $P \in \mathbb{F}_p[X]$ if $P(b)=0$.

Let $n$ be a positive integer, $p$ and $q$ be two prime numbers such that $q<n<p$. For each $a \in \mathbb{F}_q^n$ and $t \in \mathbb{F}_p$, consider the polynomial  $FP(a, \cdot) \in \mathbb{F}_p[X]$ defined as $$FP(a,t) = \sum_{i \in [n]} a_i t^{i-1}.$$  For $t \in \mathbb{F}_p$, we call $FP(a,t)$ the \emph{fingerprint} of $a$ and $t$. Note in the last expression that the coordinates of $a$ are interpreted as elements of $\mathbb{F}_p$. The following lemma is direct. Since the proof is very short we include it here.

\begin{lemma}~\cite{lidl1994introduction}\label{lem:fingerp}
Let $n$ be a positive integer, $p$ and $q$ be two prime numbers such that $q<n<p$. Let $a,b \in (\mathbb{F}_q)^n$  such that $a\neq b$. Then, $|\{ t \in \mathbb{F}_p : P(a,t) = P(b,t)\}| \leq n.$
\end{lemma}

\begin{proof}
 Note that $P(a,t) = P(b,t)$ implies that $P(a-b,t) = P(a,t)-P(b,t) = 0$. Since $P(a-b,t)$ is a polynomial of degree at most $n$ in $\mathbb{F}_p[X]$, it has at most $n$ roots in $\mathbb{F}_p$.  Therefore $|\{ t \in \mathbb{F}_p : P(a,t) = P(b,t)\}| \leq n$. 
\end{proof}

We extend the definition of fingerprints to matrices. Let $M$ be a square matrix of dimension $n$ and coordinates in  $\mathbb{F}_q$, and let $T$ be an element of $(\mathbb{F}_q)^n$. We call $FP(M, T) \in (\mathbb{F}_p)^n$ the \emph{fingerprint of $M$ and $T$}, defined as $FP(M,T) = (FP(M_1, T_1), \dots, FP(M_n, T_n))$, where $M_i$ is the $i$-th row of $M$, for each $i \in [n]$. Moreover, for a graph of size $n$, and $T\in (\mathbb{F}_p)^n$ we call $FP(G,T)$ the fingerprint of $A(G)$ and $T$.

\section{Reconstructing Hereditary Graph Classes in One Round}\label{sec:hereditary}

In this section we start giving the positive result. Later we explain the consequence of this result on well-known hereditary graph classes.

\begin{theorem}\label{theo:hereditaryexp}

Let $\mathcal{G}$ be an hereditary class of graphs. There exists a one-round private-coin algorithm that solves  {\GSRECO} whp and bandwidth $\cO(\max_{k\in [n]}(\log (|\mathcal{G}_k|)/ k) + \log n)$.

\end{theorem}

\begin{proof}
In the algorithm, nodes use a prime number $p$, whose value will be chosen later. The algorithm consists in: (1) Each node $i$ picks $t_i$ in $\mathbb{F}_p$ uniformly at random (using private coins), and computes $FP(x_i, t_i)$. (2) Each node communicates $t_i$ and $FP(x_i, t_i)$.  (3) Every node constructs $T= (t_1, \dots t_n)$ and $FP(G, T) = (FP(x_1, t_1), \dots, F(x_n, t_n))$ from the messages sent in the communication round. Finally:  (4) Every node looks in $\mathcal{G}_n$ for a graph $H$ such that $FP(H,T) = FP(G,T)$. If such graph $H$ exists, the algorithm outputs $H$, otherwise it {\it rejects}. The description of the algorithm is given in Algorithm~\ref{al:grecoheredi}.

\medskip

Now we aim to show that, if $H \in \mathcal{G}_n$ satisfies $FP(H,T) = FP(G,T)$, then $G=H$ whp. Let $T$ in $(\mathbb{Z}_p)^n$, picked uniformly at random. Then,
\begin{align*}
Pr(\exists H \in   \mathcal{G}_n \textrm{ s.t. }   H\neq G  & \textrm{ and }  FP(G,T) = FP(H,T))  \\ & \leq  \sum_{k \in [n]} Pr(\exists H \in \mathcal{G}_n \cap B(G,k) \textrm{ s.t. }   FP(G,T) = FP(H,T)).
\end{align*}
Suppose that $H \neq G$ and let $k>0$ such that $H$ belongs to $|B(G,k) \cap \mathcal{G}_n|$. Then, from Lemma~\ref{lem:fingerp}, we deduce that $Pr( FP(G,T) = FP(H,T) ) \leq \left( \frac{n}{p} \right)^k $.  It follows that  
$$Pr(\exists H \in   \mathcal{G}_n \textrm{ s.t. }  H\neq G  \textrm{ and }  FP(G,T) = FP(H,T)) \leq \left( \frac{n}{p} \right)^k \cdot |\mathcal{G}_n \cap B(G,k)|.$$

We now claim that  $|\mathcal{G}_n \cap B(G,k)| \leq {n \choose k} | \mathcal{G}_k|$.  Indeed, we can interpret a graph $H$ in $B(G,k)$ as a graph built by picking $k$ vertices $\{v_1, \dots v_k\}$ of $\mathcal{G}$ and then adding or removing edges between those vertices.  Since we are looking  for 
graphs in $|\mathcal{G}_n \cap B(G,k)|$, and $\mathcal{G}$ is hereditary, the graph induced by $\{v_1, \dots, v_k\}$ must belong to $\mathcal{G}_k$. Therefore, $|\mathcal{G}_n \cap B(G,k)| \leq {n \choose k} |\mathcal{G}_k|$. This claim implies that
\begin{align*}Pr(\exists H \in \mathcal{G}_n \textrm{ s.t. }  H\neq G  \textrm{ and }  FP(G,T) = FP(H,T))  & \leq  \sum_{k \in [n]}  \left( \frac{n^2 \cdot e \cdot (|\mathcal{G}_k|)^{1/k}}{p} \right)^k.
\end{align*}

Let $f: \mathbb{N} \rightarrow \mathbb{R}$ be defined as $f(n) = n \cdot \max_{k\in [n]} \frac{\log |\mathcal{G}_k|}{k}$. Note that this function is increasing, satisfies $f(n)/n \leq f(n+1)/(n+1)$, and $\log |\mathcal{G}_n| \leq f(n)$. 
 Therefore, $$Pr(\exists H \in \mathcal{G}_n \textrm{ s.t. }  H\neq G  \textrm{ and }  FP(G,T) = FP(H,T)) \leq \sum_{k \in [n]} \left( \frac{n^2 \cdot 2^{(f(n)/n)}}{p} \right)^k.$$

We now fix $p$ as the smallest prime number greater than  $n^3 \cdot e \cdot 2^{(f(n)/n)}$, and we deduce that 
$$Pr(\exists H \in \mathcal{G}_n \textrm{ s.t. }  H\neq G  \textrm{ and }  FP(G,T) = FP(H,T)) \leq \frac{1}{n}.$$ 

Then, with probability at least $1-1/n$, either $G=H$ or $F(H,T) \neq F(G,T)$,  for every $H\in \mathcal{G}_n$.  Hence, the algorithm solves {\GSRECO} whp.

Note that the bandwidth required by node $i$ in the algorithm equals the number of bits required to represent the pair  $(t_i, F(x_i, t_i))$, which are two integers in $[p]$. Therefore, the bandwidth of the algorithm is $$2 \lceil \log p \rceil = \cO(f(n)/n + \log n) = \cO(\max_{k\in [n]}(\log (|\mathcal{G}_k|)/ k) + \log n).$$ \end{proof}

\begin{algorithm}[htb]
\caption{{\GRECO} when $\mathcal{G}$ is hereditary. Algorithm executed by node $i$}
\label{al:grecoheredi}
\SetKwFor{Round}{Round}{}{}

	Compute $p$, the smallest prime greater than $n^3 \cdot e \cdot 2^{f(n)/n}$, where $f(n) = n\cdot \max_{k\in [n]}\frac{\log |\mathcal{G}_k|}{k}$  \;
	Pick $T_i \in \mathbb{F}_p$ uniformly at random using  private coins \;
	Compute $FP(x_i, T_i)$  \;
	Communicate $FP(x_i, T_i)$ and $T_i$  \;
	Receive $T = (T_1, \dots, T_n)$ and $FP(G,T)$ \;
	Look for $H\in \mathcal{G}_n$ such that $FP(H,T) = FP(G,T)$\;
	If $H$ exists and is unique, output $H$. Otherwise, \reject.

\end{algorithm}

\begin{corollary}\label{cor:hereditary}
Let $\mathcal{G}$ be an hereditary class of graphs, and $f$ be an increasing function such that $|\mathcal{G}_n| = 2^{\theta(nf(n))}$. Then, our private-coin algorithm solves {\GSRECO} whp, in one-round, with bandwidth $\Theta(\log|\mathcal{G}_n|/n + \log n)$. This matches the lower bound on the cost $Rb$ (which must be satisfied even in the public coin setting).
\end{corollary}

\begin{proof}
We simply note the existence of constants $c_1, c_2>0$ such that: $$\max_{k\in [n]}(\log (|\mathcal{G}_k|)/ k) \leq c_2 \cdot \max_{k\in [n]} f(k) \leq c_2 \cdot  f(n) \leq (c_2/c_1)\cdot (\log(|\mathcal{G}_n|)/n).$$ 
Therefore,  the algorithm of Theorem \ref{theo:hereditaryexp}  uses bandwidth $\cO(\log(|\mathcal{G}_n|)/n)$.
 \end{proof}

In \cite{scheinerman1994size}, Scheinerman and Zito showed that  hereditary graph classes have a very specific growing rate. They showed (\cite{scheinerman1994size}, Theorem 1) that, for any hereditary class of graphs $\mathcal{G}$,  one of the following behaviors must hold:

 \begin{itemize}
 \item $|\mathcal{G}_n|$ is \emph{constant}; meaning that $|\mathcal{G}_n| \leq 2$ for all $n$ sufficiently large. 
 \item $|\mathcal{G}_n|$ is \emph{polynomial}, meaning that $|\mathcal{G}_n| = n^{\Theta(1)}$.
 \item $|\mathcal{G}_n|$  is \emph{exponential }, meaning that $|\mathcal{G}_n| = 2^{\Theta(n)}$.
 \item $|\mathcal{G}_n|$  is \emph{factorial }, meaning that $|\mathcal{G}_n| = 2^{\Theta(n \log n)}$.
  \item $|\mathcal{G}_n|$  is \emph{super-factorial }, meaning that $|\mathcal{G}_n| = 2^{\omega(n \log n)}$.
 \end{itemize}

Corollary \ref{cor:hereditary} implies that our algorithm is tight for any factorial hereditary class of graphs. For example, the class of forests, planar graphs, interval graphs, unit disc graphs, circle graphs, etc., are factorial. Therefore, the bandwidth required to reconstruct them in one-round is~$\Theta(\log n)$. Moreover, constant, polynomial and exponential hereditary classes can be also reconstructed with bandwidth $\cO(\log n)$.

Super-factorial hereditary classes of graphs might be more troublesome. Indeed, in \cite{BALOGH2001277} it is shown that there exist super-factorial hereditary classes $\mathcal{G}$ such that the succession $\log |\mathcal{G}_n|$ might oscillate, roughly, between $cn\log n$ and $n^{1+c'}$, for two constants $c, c'>0$. For these classes,  the upper bound given by our algorithm does not match the lower bound $\Omega(\log |\mathcal{G}_n|/n)$. We remark, however, that there are also super-factorial classes of graphs where our algorithm is non-trivial and tight. For example, if $\mathcal{G}$ is the class of chordal-bipartite graphs, we have that $|\mathcal{G}_n| = 2^{\Theta(n \log^2 n)}$. Therefore, they can be reconstructed in one-round with bandwidth~$\Theta(\log^2 n)$.

%

\section{Reconstructing Arbitrary Graph Classes in Two Rounds}\label{sec:general}

In this section we show that there exists a two-round private-coin algorithm in the congested clique model that solves  {\GSRECO} whp and bandwidth $\mathcal{O}(\log|\mathcal{G}_n| /n+\log n)$.  Our algorithm is based,  roughly, on  the same ideas used to reconstruct hereditary classes of graphs. 
But the problem we encounter is the following: while in the case of hereditary classes of graphs, we had for every graph $G$ and $k>0$, a bound on the number of graphs contained in $B(G,k) \cap \mathcal{G}_n$, this is not the case in an arbitrary family of graphs $\mathcal G$. Therefore, fingerprints alone are not able to differentiate graphs. To cope with this obstacle, we use Error Correcting Codes. 

\subsection{Error Correcting Codes}

Consider the following technique, introduced by Reed and Solomon \cite{reed1960polynomial}, originally used to produce safe communication in a noisy channel. 
(This technique has also been used in randomized protocols for multiparty communication complexity \cite{fischer2016public}).

\begin{definition}
Let $0\leq k \leq n$, and let $q$ be the smallest prime number greater that $n+k$.  An \emph{error correcting code with parameters $(n, k)$} is a mapping $C: \{0,1\}^n \rightarrow (\mathbb{F}_q)^{n+k}$, satisfying:
 \begin{itemize}
  \item[1)] For every $x\in \{0,1\}^n$ and $i \in [n]$,  $C(x)_i = x_i$.
 \item[2)]  For each $x, y \in \{0,1\}^n$, $x \neq y$ implies $ |\{i \in [n+k] : C(x)_i \neq C(y)_i\}| \geq k$. 
 \end{itemize}
\end{definition}

For sake of completeness, we give the construction of an error correcting code with parameters $(n,k)$. For $x\in \{0,1\}^n$, let $P_x$ be the unique polynomial in $\mathbb{F}_q[X]$ satisfying $P_x(i) = x_i$ for each $i \in [n]$.  The function $C$ is then defined as $C(x) = (P_x(1), \dots, P_x({n+k}))$. This function satisfies both property $(1)$ from the definition of $P_x$, and property (2) because two different polynomials of degree $n$ can be equal in at most $n-1$ different values. 

We now adapt the definition of error correcting codes to graphs. 

\begin{definition}
 For a graph $G$, we call $C(G)$ the square matrix of dimension $n+k$ with elements in $\mathbb{F}_q$ defined as follows. 
\begin{itemize}
\item For each $i \in [n]$, the $i$-th row of $C(G)$ is $C(A(G)_i) \in (\mathbb{F}_q)^{n+k}$ (recall that $A(G)_i$ is the $i$-th row of the adjacency matrix of $G$). 
\item For each $i \in [k]$, the $(n+i)$-th row of $C(G)$ is the vector $(C(x_1)_{n+i}, \dots, C(x_n)_{n+i}, \vec{0}) \in (\mathbb{F}_q)^{n+k}$, where $\vec{0}$ is the zero-vector of $\mathbb{F}_q^d$, and $C(x)_{j} \in \mathbb{F}_q$ is the $j$-th element of $C(x)$.
\end{itemize}
\end{definition}

We can represent $C(x)$ as a pair $(x, \tilde{x})$, where $\tilde{x}$ belongs to $(\mathbb{F}_q)^k$. Similarly, for a graph $G$, we can represent $C(G)$ as the matrix: $$C(G) = \left[ \begin{array}{cc} A(G) & \tilde{A(G)} \\  \tilde{A(G)}^T & 0  \end{array} \right].$$
where $\tilde{A(G)}$ is the matrix with rows $C(A(G)_i)_{n+1}, \dots, C(A(G)_i)_{n+k}$, $i \in [n]$. Note that $C(G)$ is symmetric.

\medskip

\begin{remark} Note that  $d_r(C(G), C(H)) >k$, for every two different $n$-node graphs $H$ and $G$. Indeed, if $G\neq H$, there exists $i\in[n]$ such that $A(G)_i$ is different than $A(H)_i$. Then, by definition of $C$, $ |\{j \in [n+k] : C(A(G))_{i,j} \neq C(A(H))_{i,j}\}| > k$. This means that $d_r(C(G), C(H)) > k$, because $C(G)$ and $C(H)$ are symmetric matrices.
\end{remark}

\subsection{Optimal Reconstruction of Arbitrary Graph Classes in Two Rounds}

\begin{lemma}\label{lem:exT}
Let $\mathcal{G}$ be a set of graphs, $C$  the error correcting code with parameters $(n,k)$, and let $p$ be the smallest prime number greater than $(n+k) \cdot |\mathcal{G}_n|^{2/{k}}$. Then, there exists $T\in (\mathbb{F}_p)^{n+k}$ depending only on $\mathcal{G}$, satisfying $FP(C(G),T) \neq FP(C(H),T)$  for all different $G, H \in \mathcal{G}_n$.
\end{lemma}

\begin{proof}
From the remark at the end of the last subsection, we know that $d_r(C(G), C(H)) >k$, for every two different $n$-node graphs $H$ and $G$. Then, if we pick $T\in (\mathbb{F}_p)^{n+k}$ uniformly at random we have from Lemma \ref{lem:fingerp}:
$$Pr(FP(C(G),T) = FP(C(H),T)) < \left( \frac{n+k}{p}\right)^k.$$

Then, by the union bound
$$Pr(\exists G,H \in \mathcal{G}_n  \textrm{ s.t. } G\neq H  \textrm{ and } FP(C(G),T) = FP(C(H),T)) < \left( \frac{n+k}{p}\right)^k \cdot |\mathcal{G}_n|^2 \leq 1.$$

The last inequality follows from the choice of $p$. Therefore, there must exist a $T\in (\mathbb{F}_p)^{n+k}$ such that $ FP(C(G),T) \neq FP(C(H),T)$, for all different $G,H \in \mathcal{G}_n$.
\end{proof}

\begin{theorem}\label{teo:smallclassuclique} Let $\mathcal{G}$ be a set of graphs. The following holds:
\begin{itemize}

\item[1)]  There exists a two-round deterministic algorithm in the congested clique model that solves {\GRECO} with bandwidth $\cO(\log|\mathcal{G}_n| /n + \log n )$.

\item[2)]  There exists a three-round deterministic algorithm in the congested clique model that solves {\GSRECO} with bandwidth $\cO(\log|\mathcal{G}_n| /n + \log n )$.

\item[3)]  There exists a two-round private-coin algorithm in the congested clique model that solves {\GSRECO} with bandwidth $\cO(\log|\mathcal{G}_n| /n + \log n )$ whp.
\end{itemize}

\end{theorem}

\begin{proof} The first algorithm we are going to explain here,  Algorithm \ref{al:greco}, is deterministic and solves  {\GRECO} with bandwidth $\cO(\log|\mathcal{G}_n| /n + \log n)$.  The algorithms for (2) and (3) are slight modifications of Algorithm \ref{al:greco} and will also be explained in this proof.

\medskip 
 
1) Let $p$ be the first prime greater than $2n \cdot |\mathcal{G}_n|^{2/{n}}$ (then $p\leq 4n\cdot |\mathcal{G}_n|^{2/{n}}$), and let $q$ be the smallest prime number greater than $2n$.  In the algorithm, node $i$  first computes $C(x_i)$, where $C$ is the error correcting code with parameters $(n,n)$. Then, for each $j \in [n]$ node $i$ communicates $C(x_i)_{j+n}$ to node~$j$.  This communication round requires bandwidth $\lceil \log q \rceil = \cO(\log n)$.
After the first communication round, node $i$ knows $C(x_i)$ and $(C(x_1)_{i+n}, \dots, C(x_n)_{i+n})$, i.e., it knows rows $i$ and $i+n$ of matrix $C(G)$. Each node computes a vector $T  \in (\mathbb{F}_p)^{2n}$ such that $FP(C(G),T) \neq FP(C(H),T)$, for all different $G, H \in \mathcal{G}_n$ (each node computes the same $T$). The existence of $T$ is given by Lemma \ref{lem:exT}.
Then, node $i$ communicates (broadcasts) $P(C(G)_i, T_i)$ and $P(C(G)_{i+n}, T_{i+n})$. This communication round requires bandwidth $2\lceil \log p \rceil = \cO( (\log |\mathcal{G}_n|)/n + \log n )$.
After the second communication round, each node knows $P(C(G),T)$. Then,  they locally  compute the unique $H \in \mathcal{G}_n$ such that $P(C(H),T) = P(C(G),T)$. Since $G$ belongs to $\mathcal{G}_n$, then necessarily $G = H$. 

\medskip

2) Suppose now that we are solving {\GSRECO}. In this case $G$ does not necesarily belong to $\mathcal{G}_n$. After receiving the fingerprints of $C(G)$,  nodes look for a graph $H$ in $\mathcal{G}_n$ that satisfies $F(C(G),T) = F(C(H),T)$ (line 9 in Algorithm \ref{al:greco}). If such a graph exists, we call it a \emph{candidate}. Otherwise, every node decides that $G$ is not  in $\mathcal{G}_n$, so they \emph{reject}. Note that, if the candidate exists, then it is unique, since $P(C(H_1),T) \neq P(C(H_2),T)$ for all different $H_1$, $H_2$ in $\mathcal{G}_n$. 
So, if the candidate $H$ exists, each node $i$ checks whether the neighborhood of vertex $i$ on $G$ and $H$ are equal, and announces the answer in the third round (communicating one bit). If every node announces affirmatively, then they output  $G=H$. Otherwise, it means that $G$ is not  in $\mathcal{G}_n$, so every node \emph{rejects}.

\medskip

3) We now show that, if we allow the algorithm to be randomized, then we can spare the third round. In fact, nodes only need to  run Algorithm \ref{al:fingerp}  after the first round of Algorithm \ref{al:greco}. Let us explain this now. Let $p' \in [n^2, 2n^2]$ be a prime number. In the second round, node~$i$ picks $S_i \in \mathbb{F}_p$, and it communicates,  together with $FP(C(G)_i, T_{i})$ and $FP(C(G)_{i+n}, T_{i+n})$, also~$S_i$. 
After the second round of communication, if a candidate $H \in \mathcal{G}_n$ exists, each node computes $S = (S_1, \dots, S_n)$, $FP(G,S) = (FP(x_1, S_1), \dots, F(x_n, S_n)$.  If $F(G, S) = F(H,S)$, then nodes deduce that $G=H$. Otherwise, they deduce that $G\notin \mathcal{G}_n$ and \emph{rejects}.
Note that if $G$ belongs to $\mathcal{G}_n$, then the algorithm always give the correct answer. Otherwise, it rejects whp. Indeed, if $G\notin \mathcal{G}_n$, then $H\neq G$, and from Lemma \ref{lem:fingerp}, $Pr(FP(G,T) = FP(H,T)) \leq~1/n$. 
\end{proof}
\begin{algorithm}[htb]
\caption{{\GRECO}. Algorithm executed by node $i$}
\label{al:greco}
\SetKwFor{Round}{Round}{}{}

	Compute $C(x_i)$, where $C$ is the error-correcting-code with parameters $(n,n)$\;
	Communicate the element $n+j$ of $C(x_i)$ to player $j$ \;
	
	Receive $C(x_1)_{n+i}, \dots, C(x_n)_{n+i}$\;
	Call $C(x_{i+n}) = (C(x_1)_{n+i}, \dots, C(x_n)_{n+i}, \vec{0})$, where $\vec{0}$ is the zero vector of $(\mathbb{F}_p)^n$\; 
	Compute $p$ as the smallest prime greater than $2n\cdot |\mathcal{G}_n|^{2/n}$\;
	Compute $T$, the vector in $\mathbb{F}_p^{2n}$, given by Lemma \ref{lem:exT} \;
	Compute and communicate (broadcast) $FP(C(x_i), T_i)$ and $FP(C(x_{n+i}), T_{n+i})$\;
	Receive $FP(C(G),T)$\;
	Look for $H\in \mathcal{G}_n$ such that $FP(C(H),T) = FP(C(G),T)$\;
	Output $H$. 

\end{algorithm}

\vspace{-0.25cm}
\begin{algorithm}[htb]
\caption{Checking a candidate $H$. Algorithm executed by node $i$}
\label{al:fingerp}
\SetKwFor{Round}{Round}{}{}

	Compute $p'$, the smallest prime number such that $p'>n^2$\;
	Pick $T_i \in \mathbb{F}_p$ uniformly at random using private coins \;
	Compute $FP(x_i, T_i)$  \;
	Communicate $FP(x_i, T_i)$ and $T_i$  \;
	Receive $T = (T_1, \dots, T_n)$ and $FP(G,T)$ \;
	Output $H$ if $FP(H,T) = FP(G,T)$, otherwise \reject .
\end{algorithm} 

Note that our private-coin algorithm for {\GSRECO} has one-sided error. In fact, if the input graph belongs to $\mathcal{G}$, then our algorithm reconstructs it with probability $1$. On the other hand, if $G$ is not contained in $\mathcal{G}$, then our algorithm fails to discard the candidate with probability at most $1/n$. 

%
%
%
%

\section{Revisiting the One Round Case}\label{sec:revisiting}

In this section we revisit the one-round case (and therefore the broadcast congested clique model). But instead of studying hereditary graph classes
we study arbitrary graph classes, and we show that for this general case we need a larger bandwith.  Our results are tight, not only in terms of the bandwidth,
but also in the necessity of using randomization.


\begin{theorem}\label{theo:general1R} Let $\mathcal{G}$ be a set of graphs. The following holds:
\begin{itemize}

\item[1)]  There exists a one-round deterministic algorithm in the congested clique model that solves {\GRECO} with bandwidth $\cO(\sqrt{ \log |\mathcal{G}_n| \log n} + \log n )$.

\item[2)]  There exists a two-round deterministic algorithm in the broadcast congested clique model that solves {\GSRECO} with cost $\cO(\sqrt{ \log |\mathcal{G}_n| \log n} + \log n )$.

\item[3)]  There exists a one-round private-coin algorithm in the congested clique model that solves {\GSRECO} with bandwidth $\cO(\sqrt{ \log |\mathcal{G}_n| \log n} + \log n )$ whp.
\end{itemize}
\end{theorem}

\begin{proof}
The algorithm in this case is very similar to the one we provided in the proof of Theorem \ref{teo:smallclassuclique}. Let $k$ be a parameter
whose value will be chosen at the end of the proof,  and let $C$ be the error-correcting-code with parameters $(n,k)$. Let $p$ be the smallest prime number greater than $2n\cdot |\mathcal{G}|^{2/k}$. Let $T \in (\mathbb{F}_p)^{n+k}$ be the vector given by Lemma \ref{lem:exT}, corresponding to $\mathcal{G}$. 

In the algorithm, every node $i$ computes $C(x_i)$, and communicates $FP(C(x_i), T_i)$ together with $C(x_i)_{n+1}, \dots, C(x_i)_{n+k} \in (\mathbb{F}_q)^k$, where $q$ is the smallest prime greater than $k+n$. Note that the communication round requires bandwidth $$\cO( \log p  + k \cdot \log(n+k)) = \cO( \log |\mathcal{G}_n| / k + (k+1)\cdot \log n).$$

After the communication round, every node knows $FP(C(x_i),T_i)$, for all $i\in [n]$, and also knows the matrix $\tilde{A(G)}$. Therefore, every node can compute $F(C(x_i), T_i)$, for all $i\in \{n+1, \ldots,n+k\}$, and, moreover, compute $F(C(G),T)$. 
 
From the construction of $T$, there is at most one graph $H \in \mathcal{G}_n$ such that $F(C(G),T) = F(C(H),T)$. Therefore, if $G$ belongs to $\mathcal{G}$, every node can reconstruct it. On the other hand, if we are solving {\GSRECO}, then we proceed as in the algorithm of Theorem \ref{teo:smallclassuclique}, either testing whether $H=G$ in one more round, or sending a fingerprint of $G$ to check with high probability if a candidate $H\in \mathcal{G}_n$ such that $F(C(G),T) = F(C(H),T)$ is indeed equal to $G$. This verification requires to send $\cO(\log n)$ more bits, which fits in the asymptotic bound of the bandwidth.

The optimal value of $k$, that is, the one which minimizes the bandwidth, is such that  $k = \cO\left(\sqrt{\frac{\log |\mathcal{G}_n|}{\log n}}\right)$. 
Threfore, the bandwidth is $\cO(\sqrt{ \log |\mathcal{G}_n| \log n} + \log n )$. 
\end{proof}


%
%

\subsection{Tightness of our Algorithms}

In this subsection we show that our algorithms for solving {\GRECO} and {\GSRECO} are tight, from two different perspectives. First, from the point of view of the bandwidth, we show that there are classes of graphs $\mathcal{G}$ satisfying $|\mathcal{G}_n| \leq 2^{\cO(n)}$ such that every algorithm (deterministic or randomized) solving {\GRECO} in the broadcast congested clique model has cost $Rb=\Omega(\sqrt{\log |\mathcal{G}_n|})$. This lower bound matches the upper {\emph{one-round}} bound given in Theorem \ref{theo:general1R} (up to logarithmic factors). 

Then, we show  that, when  restricted to one-round algorithms, the use of randomization is necessary in order  to have non-trivial general algorithms solving {\GSRECO}. Indeed, we prove that there exists a set of graphs $\mathcal{G}$ satisfying $|\mathcal{G}_n|\leq 2^n$ such that, every one-round deterministic algorithm that solves  {\GSRECO}, requires bandwidth $\Omega(n)$.

\begin{theorem}
There exists a class of graphs $\mathcal{G}$ satisfying   $|\mathcal{G}_n|\leq 2^{\cO(n)}$  such that,  any $\epsilon$-error public-coin  
algorithm in the broadcast congested clique model that solves  {\GRECO}, has cost $Rb = \Omega(\sqrt{n}) = \Omega(\sqrt{\log |\mathcal{G}_n|})$. 

\end{theorem}
\newpage
\begin{proof}
Let $\mathcal{G}^+$ be the class of graphs defined as follows: $G$ belongs to $\mathcal{G}_n^+$ if and only if $G$ is the disjoint union of a graph $H$ of $\lceil \sqrt{n} \rceil$ nodes and $n - |H|$ isolated nodes. Note that $|\mathcal{G}^+_n| = {n \choose \lceil \sqrt{n} \rceil} \cdot 2^{{\lceil \sqrt{n} \rceil} \choose 2} \leq 2^{\cO(n)}$. Indeed, there are $2^{{\lceil \sqrt{n} \rceil} \choose 2} = 2^{\cO(n)}$ labeled graphs of size ${\lceil \sqrt{n} \rceil}$, and at most ${n \choose \lceil \sqrt{n} \rceil} = 2^{\cO(\sqrt{n}\log n)}$ different labelings of a graph of $\sqrt{n}$ nodes using $n$ labels (so $\mathcal{G}^+$ is closed under isomorphisms). 

Let  $\mathcal{A}$ be an $\epsilon$-error public-coin algorithm  solving $\mathcal{G}^+$-{\sc Weak-Rec} in $R(n)$ rounds and bandwidth $b(n)$, on input graphs of size $n$. 

Consider now the following algorithm $\mathcal{B}$ that solves $\mathcal{U}$-{\sc Weak-Rec}, where $\mathcal{U}$ is the set of all graphs: on input graph $G$ of size $n$, each node $i \in [n]$ supposes that it is contained in a graph $G^+$ formed by $G$ plus $n^2-n$ isolated vertices with identifiers $(n+1), \dots, n^2$. Note that $G^+$ belongs to $\mathcal{G}^+$. Then, node $i$ simulates $\mathcal{A}$  as follows: at each round, node $i\in [n]$ produces the message of node $i$ in $G^+$ according to $\mathcal{A}$. Note that the messages produced by nodes labeled $(n+1), \dots, n^2$ do not depend on $G$, so they can be produced by any node of~$G$.  Since $\mathcal{A}$ solves $\mathcal{G}^+$-{\sc Weak-Rec}, at the end of the algorithm every node knows all the edges of $G^+$, so they reconstruct $G$ ignoring vertices labeled $(n+1), \dots, n^2$. 

We deduce that algorithm $\mathcal{B}$ solves $\mathcal{U}$-{\sc Weak-Rec}. Note that the cost of $\mathcal{B}$ is $n^2R(n)b(n)$ on input graphs of size $n$. We deduce that $n^2R(n)b(n)= \Omega(n)$, i.e., the cost of $\mathcal{A}$ is $\Omega(\sqrt{n})$.  
\end{proof}

We say that an algorithm \emph{recognizes} $\mathcal{G}$ if the algorithm decides whether an input graph $G$ belongs to $\mathcal{G}$. We call {\GRECON} the problem of recognizing  $\mathcal{G}$.

\begin{theorem}
There exists a set of graphs $\mathcal{G}$ satisfying  $|\mathcal{G}_n|\leq 2^n$ such that, and any one-round deterministic algorithm 
in the congested clique model that solves {\GRECON}, requires bandwidth $\Omega(n)$.
\end{theorem}

\begin{proof}
We prove this theorem by a counting argument. Our goal is to show that there are more \emph{small} sets of graphs than one-round deterministic algorithms capable to recognize them.

 We first count the number of sets of graphs (not necessarily closed under taking isomorphism) containing $2^{n}$ different graphs of size $n$. We call the family of these sets  $\mathcal{C}$. There are $2^{n \choose 2}$ possible graphs of size $n$, so $2^{n \choose 2} \choose 2^{n}$ possible choices for graphs in $\mathcal{C}$. We deduce that there exists $c_1>0$ such that $|\mathcal{C}| \geq  2^{c_1\cdot n^2 \cdot 2^{n}}$.

On the other hand, we count the number of one-round deterministic algorithms that recognize a set of graphs in $\mathcal{C}$ with bandwidth at most $\beta$.  A one-round deterministic algorithm is composed of two parts: the algorithm before the communication round, and the algorithm after the communication. The first part of an algorithm is defined by the messages that a node sends on each input. The input of a node is its neighborhood represented by a Boolean vector of size $n$, and an integer representing its label. Therefore, the first part of an algorithm is defined by the messages corresponding to all the $n2^{n}$ possible inputs. Since the bandwidth is $\beta$, we obtain that there are $2^{n\beta 2^n}$ possible choices for the first part of an algorithm. 


The second part of an algorithm is defined by a function $f_\mathcal{G}: (\{0,1\}^{b})^n \rightarrow \{0,1\}$, such that if $m = (m_1, \dots, m_n)$ are the messages sent by the nodes in the communication round, then $f(m) = 1$ if and only if $m$ was produced from an input graph belonging to $\mathcal{G}$. The crucial observation is that this implies that $f$ can output $1$ in at most $2^{n}$ inputs. Therefore, the number of possible second parts of an algorithm is $\sum_{i\in [2^n]}{2^{n\beta} \choose i} \leq (1+2^{n\beta})^{2^n} \leq 2^{c_2 \cdot n\beta2^n} $, where $c_2>0$ is a constant.

We deduce that the number of one-round deterministic algorithms with bandwidth $\beta$ that are capable to recognize a set of graphs in $\mathcal{C}$ is at most $2^{c_3 n\beta 2^n }$, with $c_3>0$. Since we are considering only deterministic algorithms, two different sets must be recognized by two different algorithms. This implies that $2^{c_3n\beta 2^n}$ must be greater than $2^{c_1n^2 2^{n}}$, so $\beta = \Omega(n)$. 

Finally, we construct $\mathcal{G}$ by picking, for each $n$, one set of graphs contained in $\mathcal{C}$ that can not be recognized by any algorithm of bandwidth $o(n)$. 
\end{proof}

\begin{remark}
Note that  for any set of graphs $\mathcal{G}$, problem {\GSRECO} is at least as hard as {\GRECON}. We conclude that there exists a set of graphs $\mathcal{G}$ satisfying  $|\mathcal{G}_n|\leq 2^n$ such that,   any one-round deterministic algorithm  that solves {\GSRECO}, requires bandwidth $\Omega(n)$. Note that,  since in this case $|\mathcal{G}_n|\leq 2^n$, from  Theorem \ref{theo:general1R} we know that 
 {\GSRECO}  can be solved using a one-round private-coin algorithm with bandwidth $\cO(\sqrt{n\log n})$ whp. 
\end{remark}

\vspace{-0.25cm}
\section{Discussion}\label{sec:conclusion}

In this paper we showed  that all graph classes can be optimally reconstructed in two rounds in the congested clique model. 
But our algorithm is randomized, it uses private-coins. A natural question is the following: is it possible to achieve the same deterministically?
In other words, given an arbitrary graph class $\mathcal{G}$, is it always possible to solve {\GSRECO} with a two-round deterministic algorithm
with bandwidth  $\cO(\log |\mathcal{G}_n|/n + \log n)$? (Note that this is true for the weak version of the reconstruction problem {\GRECO}).

We also restricted the reconstruction problem to one-round algorithms. We showed that, if  $\mathcal{G}$ is an hereditary graph class such as forests, planar graphs, interval graphs, unit disc graphs, chordal bipartite graphs, bounded tree-widh graphs, $d$-degenerate graphs, etc., then 
{\GSRECO} can be solved, whp, with a one-round private-coin
algorithm that uses  bandwidth  $\cO(\log |\mathcal{G}_n|/n)$. Can we extend this result to every hereditary class of graphs?

A related problem is the recognition problem, where we simply want to decide whether the input graph belongs to the class $\mathcal{G}$. It seems that sometimes we can not solve the recognition problem without solving the reconstruction problem. This seems to be true in the case of trees and, more generally, in the case of $d$-degenerate graphs. But this is not always the case. Sometimes,  solving the recognition problem requires a much smaller bandwidth. For example, consider the class of split graphs. A split graph is a graph where the vertices can be partitioned into a clique and an independent set (these two sets are connected arbitrarily). The class of split graphs contains $2^{\Omega(n^2)}$ graphs of size $n$, so it cannot be reconstructed with cost $o(n)$. However, split graphs can be characterized solely by their degree sequences  (see \cite{brandstadt1999graph}), so they can be recognized by a one-round  deterministic algorithm, where each node sends its degree ($\cO(\log n)$ bits). 
It is an interesting challlenge to understand the cases where  we can solve the recognition problem {\emph{without}} solving the reconstruction problem.

\vspace{-0.25cm}

\bibliographystyle{plain}

\end{document}